\documentclass[copyright,creativecommons]{eptcs}
\providecommand{\event}{GandALF 2014} 
\usepackage{breakurl}             
\usepackage{amsfonts}
\usepackage{amssymb}
\usepackage{amsmath}
\usepackage{amsthm}
\usepackage{verbatim}

\newtheorem{theorem}{Theorem}[section]
\newtheorem{lemma}[theorem]{Lemma}

\title{Infinite Networks, Halting and Local Algorithms}
\author{Antti Kuusisto\thanks{The author acknowledges that
this work was carried out during a tenure of the
ERCIM ``Alain
Bensoussan" Fellowship Programme.
The reported research has received funding from the
European Union
Seventh Framework
Programme (FP7/2007-2013) under
grant agreement
number 246016.}
\institute{Institute of Computer Science\\ University of Wroc\l aw\\ Poland}
\email{antti.j.kuusisto@uta.fi}
}
\def\titlerunning{Infinite Networks, Halting and Local Algorithms}
\def\authorrunning{A. Kuusisto}
\begin{document}
\maketitle

\begin{abstract}
The immediate past has witnessed an increased amount of interest in local algorithms, i.e., constant time distributed algorithms. In a recent survey of the topic (Suomela, ACM Computing Surveys, 2013), it is argued that local algorithms provide a natural framework that could be used in order to theoretically control infinite networks in finite time. We study a comprehensive collection of distributed computing models and prove that if infinite networks are included in the class of structures investigated, then every universally halting distributed algorithm is in fact a local algorithm. To contrast this result, we show that if only finite networks are allowed, then even very weak distributed computing models can define nonlocal algorithms that halt everywhere. The investigations in this article continue the studies in the intersection of logic and distributed computing initiated in (Hella et al., PODC 2012) and (Kuusisto, CSL 2013).
%

%
%
%
%
\end{abstract}

\section{Introduction}

This work is a study of deterministic distributed algorithms for arbitrary networks,
including infinite structures in addition to finite ones.
In the recent survey article \cite{suomela}, Suomela points out that distributed constant-time algorithms are a reasonable
choice for theoretically controlling infinite networks in finite time. In this article we show that for a rather comprehensive
collection of models of distributed computing, constant-time algorithms are in a
sense the \emph{only} choice. We define a
framework---based on a class of message passing automata and relational structures---that contains a
comprehensive variety of models of distributed computing in \emph{anonymous networks}, i.e., networks
without $\mathrm{ID}$-numbers.
We then show that if infinite networks are allowed,
then \emph{all universally halting algorithms} definable in the framework are in fact local algorithms, i.e.,
distributed constant-time algorithms. 
The widely studied \emph{port-numbering model} (see \cite{angluin, hella, hella2})
of distributed computing can be directly extended to a framework that contains infinite structures
in addition to finite ones.
In the port-numbering model, a node of degree $k\leq n$, where $n$ is a globally known
finite degree bound, receives messages through $k$ input ports
and sends messages through $k$ output ports.
The processors in the nodes can send different messages to different neighbours,
and also see from which port incoming messages arrive.
There are no $\mathrm{ID}$-numbers 
in this framework. The omission of $\mathrm{ID}$-numbers is well justified when infinite
networks are studied: in most natural theoretical frameworks for the
modelling of computation in infinite networks,
\emph{even the reading of all local $\mathrm{ID}$s in the 
beginning of computation would take infinitely long}. Thus typical synchronized
communication using $\mathrm{ID}$-numbers would be impossible.
There are several fields of study outside distributed computing
where the objects of investigation can be regarded as infinite
distributed anonymous communication networks.
\emph{Cellular automata}
%
%
provide probably the most
obvious and significant example of such a framework.
But of course there are various others.
Crystal lattices and the brain, for example, are 
massive network systems often modelled by infinite structures.
%
%
%

%
Below we define a \emph{general  
distributed computing model} based on relational structures and synchronized message passing automata.
The port-numbering model $\mathrm{VV_c}$ of \cite{hella, hella2} and all its subclasses can be directly
simulated in our framework by restricting attention to suitable
classes of structures and automata.
%
%
We establish (Theorem \ref{maintheorem}\hspace{0.4mm}) that if $\mathcal{H}$ is a class of
communication networks definable by a 
\emph{first-order theory}, then all universally halting algorithms over $\mathcal{H}$ are local algorithms.
For example, the classes of networks for the $\mathrm{VV}_c$ model are easily seen to be definable by
first-order formulae, \emph{as long as infinite structures are allowed}.
In fact, when the requirement of finiteness is lifted,
all classes of structures in the comprehensive study
in \cite{hella, hella2} can easily be seen to be first-order definable.
The proof of Theorem \ref{maintheorem} makes a \emph{crucial use of logic},
thereby extending the work initiated in \cite{hella, hella2}
and developed further in \cite{kuusi}.
The articles \cite{hella, hella2, kuusi} extend the scope of \emph{descriptive complexity theory}
(see \cite{ebbinghaus, Immerman, libkin}) to the realm of distributed computing by identifying a
highly canonical one-to-one link between \emph{local} algorithms and
formulae of modal logic.
This link is based on the novel idea of 
directly identifying \emph{Kripke models}
and distributed communication networks with each other.
Under this interpretation, arrows of the accessibility relations of
Kripke models are considered to \emph{\textbf{be}} communication channels
between processors in distributed networks.
This idea has turned out to be fruitful because it enables the
transfer of results between modal logic and distributed computing.
For example, in \cite{hella, hella2} a novel separation argument
concerning distributed complexity classes is obtained by applying the
\emph{bisimulation method} (see \cite{johan, han, blackburn}) of modal logic to
distributed communication networks.
In this article we adapt the link between modal logic
and distributed computing for the purpose of
proving Theorem \ref{maintheorem}\hspace{0.4mm}.
We first obtain a characterization of 
halting behaviour in terms of modal formulae.
This facilitates the use of the \emph{compactness theorem} (see \cite{ebbinghaus}),
which is the final step in our  proof.
To contrast Theorem \ref{maintheorem}\hspace{0.4mm}, we investigate halting
behaviour of distributed message passing automata in the finite.
We establish that even extremely weak subsystems of the port-numbering model
can define nonlocal halting algorithms when attention is restricted to
finite networks:
Theorem \ref{secondmain} shows that even if message passing automata in the port-numbering
model have absolutely no access to port numbers whatsoever,
nonlocal but universally halting behaviour is possible.

In order to prove Theorem \ref{secondmain}\hspace{0.4mm}, we employ tools from
\emph{combinatorics on words}, namely, the 
infinite \emph{Thue-Morse sequence} (see \cite{allouche}). This infinite binary sequence is known to be cube-free,
i.e., it does not have a prefix of the type $tuuu$, where $u$ is a nonempty word. This 
lack of periodicity allows us to design an appropriate algorithm that is halting but nonlocal
in the finite.
\section{Preliminaries}\label{preliminaries}
%
%

%
%

%
Let $\Pi$ be a finite  set of
\emph{unary relation symbols} $P\in\Pi$ and $\mathcal{R}$ a finite set of
\emph{binary relation symbols} $R\in\mathcal{R}$. These symbols are also called
\emph{predicate symbols}. 
The set of $(\Pi,\mathcal{R})$-\emph{formulae} of \emph{modal logic} $\mathrm{ML}(\Pi,\mathcal{R})$
is generated by the grammar
$$\varphi\ ::=\ \top\ |\ P\ |\ \neg\varphi\ |\ (\varphi_1\wedge\varphi_2)\ |\ \langle R\, \rangle \varphi,$$
where $P$ is any symbol in $\Pi$,  $R$ any symbol in $\mathcal{R}$, and $\top$ is a logical constant  symbol.
Let $\mathrm{VAR} = \{\ x_i\ |\ i\in\mathbb{N}\ \}$ 
be a set of \emph{variable symbols}. 
The set of $(\Pi,\mathcal{R})$-\emph{formulae} of \emph{first-order logic} $\mathrm{FO}(\Pi,\mathcal{R})$ 
is generated by the grammar
$$\varphi\ ::=\ \top\ |\ x = y\ |\ 
P(x)\ |\ R(x,y)\ |\ \neg\varphi\ |\ (\varphi_1\wedge\varphi_2)\ |\ \exists x\, \varphi,$$
where $x$ and $y$ are symbols in $\mathrm{VAR}$,
$P$ a symbol in $\Pi$,  $R$ a symbol in $\mathcal{R}$, and $\top$ a logical constant  symbol.
For both logics, we define the abbreviation $\bot := \neg\top$. We also
use the abbreviation symbols $\vee$, $\rightarrow$ and $\leftrightarrow$ in the usual way.
The \emph{modal depth} $\mathit{md}(\varphi)$ of a formula is defined recursively such that
$\mathit{md}(\top) = \mathit{md}(P) = 0$, $\mathit{md}(\neg\, \psi) = \mathit{md}(\psi)$,
$\mathit{md}(\psi\wedge\chi) = \mathit{max}\{\mathit{md}(\psi),\mathit{md}(\chi)\}$, and
$\mathit{md}(\langle R\, \rangle\, \psi) = \mathit{md}(\psi) + 1$.
Let $\Pi = \{P_1,...,P_n\}$ and $\mathcal{R} = \{R_1,...,R_m\}$.
A $(\Pi,\mathcal{R})$-\emph{model} 
is a structure
$$M\ :=\ (W,P_1^M,...,P_n^M, R_1^M,...,R_m^M),$$
where $W$ is an arbitrary nonempty set (the \emph{domain}
of the model $M$), 
each $P_i^M$ is a unary relation $P_i^M\subseteq W$, and
each $R_i^M$ a binary relation $R_i^M\subseteq W\times W$.
The semantics of $\mathrm{ML}(\Pi,\mathcal{R})$ is defined
with respect to \emph{pointed $(\Pi,\mathcal{R})$-models}
$(M,w)$, where $M = (W,P_1^M,...,P_n^M,R_1^M,...,R_m^M)$
is a $(\Pi,\mathcal{R})$-model and $w\in W$ a \emph{point} or a \emph{node} of (the
domain of) $M$.
For each $P_i\in\Pi$, we define $(M,w)\models P_i$ iff $w\in P_i^M$. We also define $(M,w)\models \top$.
We then recursively define
\[
\begin{array}{lll}
(M,w)\models \neg\, \varphi\ \ &\Leftrightarrow\ \ \ (M,w)\not\models\varphi,\\
(M,w)\models (\varphi\wedge\psi)\ \ &\Leftrightarrow\ \ \ (M,w)\models \varphi\text{ and }(M,w)\models\psi,\\
(M,w)\models \langle R_i\, \rangle\, \varphi\ \ &\Leftrightarrow\
\ \ \ \exists v\in W\bigl(\, (w,v)\in R_i^M\text{ and }(M,v)\models\varphi\,\bigr).
\end{array}
\]
The semantics of $\mathrm{FO}(\Pi,\mathcal{R})$ is defined
in the usual way with respect to \emph{$(\Pi,\mathcal{R})$-interpretations}
$(M,f)$, where
$$M\ =\ (W,P_1^M,...,P_n^M, R_1^M,...,R_m^M)$$
is a $(\Pi,\mathcal{R})$-model and $f$ is an \emph{assignment function}
$f:\mathrm{VAR}\rightarrow W$ giving an interpretation to the variables in $\mathrm{VAR}$.
We define $(M,f)\models x = y\, \Leftrightarrow\, f(x) = f(y)$, $(M,f)\models P_i(x)\, \Leftrightarrow\, f(x)\in P_i^M$,
and $(M,f)\models R_i(x,y)\, \Leftrightarrow\, (f(x),f(y))\in R_i^M$. We also define $(M,f)\models\top$.
We then recursively define
\[
\begin{array}{lll}
(M,f)\models \neg\, \varphi\ \ &\Leftrightarrow\ \ \ (M,f)\not\models\varphi,\\
(M,f)\models (\varphi\wedge\psi)\ \ &\Leftrightarrow\ \ \ (M,f)\models \varphi\text{ and }(M,f)\models\psi,\\
(M,f)\models \exists x\varphi\ \ &\Leftrightarrow\ \ \ \ \exists v\in W\bigl(\, (M,f[x\mapsto v])\models\varphi\,\bigr),
\end{array}
\]
where $f[x\mapsto v]$ is the function $g:\mathrm{VAR}\rightarrow W$ such that
\[
 g(y) = \begin{cases}
    v & \text{ if } y=x,\\
    f(y) & \text{ if } y\not=x.
 \end{cases}
\]
It is well known that modal logic can be directly translated into first-order logic.
We define the \emph{standard translation} from $\mathrm{ML}(\Pi,\mathcal{R})$ into $\mathrm{FO}(\Pi,\mathcal{R})$ in
the following way. We let $\mathit{St}_x(\, \top )\, :=\, \top$, $\mathit{St}_x(P_i )\, :=\, P_i(x)$,
$\mathit{St}_x(\, (\varphi\wedge\psi)\, )\, :=\, \bigl(\mathit{St}_x(\varphi )\wedge\mathit{St}_x(\psi)\bigr)$,
$\mathit{St}_x(\neg\varphi)\, :=\, \neg\mathit{St}_x(\varphi)$, and $\mathit{St}_x(\langle R_i\, \rangle\varphi) := 
\exists y\bigl( R_i(x,y)\wedge \mathit{St}_y(\varphi)\bigr)$.
Here $y$ is a fresh variable distinct from $x$.
It is easy to see that $(M,w)\models\varphi$ iff
$(M,f[x\mapsto w])\models \mathit{St}_x(\varphi)$. 
Due to the standard translation, modal logic is often 
considered to be simply a \emph{fragment} of first-order logic.
We next fix some conventions concerning \emph{sets}
of formulae. We only discuss formulae of first-order logic,
but analogous definitions hold for modal logic.
If $\Phi$ is a set of formulae of $\mathrm{FO}(\Pi,\mathcal{R})$,
then $\bigvee\Phi$ and $\bigwedge\Phi$ denote the \emph{disjunction} and
\emph{conjunction} of the formulae in $\Phi$. The set $\Phi$ can be infinite, but 
then of course neither $\bigvee\Phi$ nor $\bigwedge\Phi$ is a formula of 
%
%
$\mathrm{FO}(\Pi,\mathcal{R})$.
%
%
We define $(M,f)\models\bigvee\Phi$ if there exists at least one formula $\varphi\in\Phi$ such that $(M,f)\models\varphi$.
We define $(M,f)\models\bigwedge\Phi$ if $(M,f)\models\varphi$ for all $\varphi\in\Phi$.
A \emph{set} of formulae of $\mathrm{FO}(\Pi,\mathcal{R})$
is called a \emph{theory} (over the
signature $(\Pi,\mathcal{R})$).\hspace{0.4mm}\footnote{A theory does not have to be closed under
logical consequence. A theory is simply a set of formulae, and
can be infinite or finite.} If $T$ is a theory over the signature $(\Pi,\mathcal{R})$,
then $(M,f)\models T$ means that $(M,f)\models\varphi$ for all $\varphi\in T$.
%
%
%
%
%
When we write $T\models\varphi$, we mean that
the implication $(M,f)\models T\Rightarrow (M,f)\models\varphi$ holds
for all $(\Pi,\mathcal{R})$-interpretations $(M,f)$.
As usual, two $\mathrm{FO}(\Pi,\mathcal{R})$-formulae $\varphi$ and $\psi$
are \emph{equivalent} if the equivalence $(M,f)\models \varphi\Leftrightarrow (M,f)\models \psi$ holds
for all $(\Pi,\mathcal{R})$-interpretations $(M,f)$.
%
%
%
%

%
Let $\mathcal{H}$ be a class of \emph{pointed} $(\Pi,\mathcal{R})$-models, and let $\mathcal{K}\subseteq\mathcal{H}$.
A modal formula $\varphi$ \emph{defines} the class $\mathcal{K}$ \emph{with respect to} $\mathcal{H}$, if
for all $(M,w)\in\mathcal{H}$, we have $(M,w)\models\varphi\ \Leftrightarrow\ (M,w)\in\mathcal{K}$.
%
%
If some formula $\psi$ defines a class $\mathcal{J}$ of pointed $(\Pi,\mathcal{R})$-models with respect to the
class of all pointed $(\Pi,\mathcal{R})$-models, we simply say that $\psi$ defines $\mathcal{J}$.
%

%
%
%
%
%
%
%
If $\varphi$ is a \emph{sentence} of $\mathrm{FO}(\Pi,\mathcal{R})$ 
and $M$ a $(\Pi,\mathcal{R})$-model,
we write $M\models\varphi$ if $(M,f)\models\varphi$
for some assignment $f$. (Trivially, whether $(M,f)\models\varphi$ holds or not, does not
depend on $f$ when $\varphi$ is a sentence.)
If $T$ is a theory consisting of $\mathrm{FO}(\Pi,\mathcal{R})$-sentences,
we write $M\models T$ iff $M\models\psi$ for all $\psi\in T$.
Let $\mathcal{J}$ be a class of \emph{pointed} $(\Pi,\mathcal{R})$-models
and $T$ a theory consisting of \emph{$\mathrm{FO}(\Pi,\mathcal{R})$-sentences}.
We say that the \emph{first-order theory $T$ defines the class $\mathcal{J}$
of pointed models} if for all pointed $(\Pi,\mathcal{R})$-models $(M,w)$, we have
$M\models T\Leftrightarrow (M,w)\in\mathcal{J}$.
Notice indeed that accoring to this convention, if $T$
defines a class $\mathcal{J}$ of pointed models
and if $w$ is a point in
the domain of $M$ and $(M,u)$ a pointed model in $\mathcal{J}$, then
we have $(M,w)\in\mathcal{J}$.
If a first-order theory $T$ defines a class $\mathcal{H}$ of pointed models,
then we say that $\mathcal{H}$ is \emph{definable} by
the first-order theory $T$.
If $\mathcal{H}$ is definable by a theory $\{\varphi\}$ containing a
single first-order $(\Pi,\mathcal{R})$-sentence $\varphi$, we say that $\mathcal{H}$
is definable by the first-order sentence $\varphi$.
Let $\Pi$ and $\mathcal{R} = \{\, R_1,...,R_k\, \}$ be finite sets of unary and
binary relation symbols, respectively.
A \emph{message passing automaton} $A$ over the signature 
$(\Pi,\mathcal{R})$, or a $(\Pi,\mathcal{R})$-automaton,
is a tuple
$$(Q,\mathcal{M},\pi,\delta,\mu,F,G)$$
defined as follows.
$Q$ is a nonempty set of \emph{states}. $Q$ can be finite or countably infinite.
$\mathcal{M}$ is a nonempty set of \emph{messages}.
$\mathcal{M}$ can be finite or countably infinite.
For a set $S$, we let $\mathit{Pow}(S)$ denote the power set of $S$.
$\pi:\mathit{Pow}(\Pi)\rightarrow Q$ is an \emph{initial transition function} that
determines the beginning state of the automaton $A$. $\delta:((\mathit{Pow}(\mathcal{M}))^k\times Q)\ \rightarrow\ Q$
is a \emph{transition function} that constructs a new state $q\in Q$
when given a $k$-tuple $(N_1,...,N_k)\in(\mathit{Pow}(\mathcal{M}))^k$ of received message sets
and the current state. $\mu:(Q\times\mathcal{R}) \rightarrow \mathcal{M}$ is a \emph{message construction function}
that constructs a message for the automaton to send forward when given the current state of the automaton and a
\emph{communication channel} $R_i\in\mathcal{R}$.
$F\subseteq Q$ is the set of \emph{accepting states} of the automaton.
$G\subseteq Q\setminus F$ is the set of \emph{rejecting states} of the automaton.
Let $\mathcal{R} = \{\, R_1,...,R_k\, \}$ and $\Pi = \{\, P_1,...,P_m\, \}$.
Let $(M,w)$ be a $(\Pi,\mathcal{R})$-model. The set of $R_i$-\emph{predecessors} of $w$ is the set of
nodes $u$ in the domain of $M$ such that $R_i(u,w)$, and the set of $R_i$-\emph{successors} of $w$ is
the set of nodes $u$ such that $R_i(w,u)$. The set of $R_i$-successors of $w$ is denoted by $\mathit{succ}(R_i,w)$.
%
%
%

%
A message passing $(\Pi,\mathcal{R})$-automaton $A$ 
is \emph{run} on a $(\Pi,\mathcal{R})$-model $M = \bigl(W,R_1,...,R_k,P_1,...,P_m\bigr)$,
considered to be a distributed system. We first give an intuitive
description of the computation of the distributed system defined by
the automaton $A$ and the model $M$, and then define the computation
procedure more formally.
On the intuitive level, we place a copy $(A,w)$ of the automaton $A$ to each node $w\in W$.
Then, each automaton $(A,w)$  first scans the \emph{local  information} of the node $w$, i.e.,
finds the set of unary relation symbols $P_i\in\Pi$ such that $(M,w)\models P_i$, and then makes a transition to
a \emph{beginning state} based on the local information.
The local information at $w$ can be considered to be an $m$-bit string $t$ of zeros and ones such
that the $i$-th bit of $t$ is $1$ iff $(M,w)\models P_i$.
After scanning the local information, the automata $(A,w)$, where $w\in W$,
begin running in \emph{synchronized steps}. During
each step, each automaton $(A,w)$ sends, for each $i\in\{\, 1,...,k\, \}$, a message $m_i$ to the
$R_i$-\emph{predecessors} of $w$.\hspace{0.4mm}\footnote{Therefore information flows opposite to the direction of
the arrows (i.e., ordered pairs) of $R_i$. The reason for this choice is technical, and could be avoided.
The choice is due to the relationship between modal logic and message passing automata.
A possible alternative approach would be to consider modal logics with the truth of $\langle R_i\, \rangle\varphi$
defined such that $(M,w)\models \langle R_i\, \rangle\varphi$
iff $\exists v\in W\bigl(\, (v,w)\in R_i^M\text{ and }(M,v)\models\varphi\,\bigr)$.}
The automaton $(A,w)$ also receives a tuple $(N_1,...,N_k)$ of message sets $N_i$ such that
set $N_i$ is received from 
the $R_i$-successors of $w$. Then the automaton updates its state based on the received messages and
the current state.
More formally, a $(\Pi,\mathcal{R})$-model $\bigl(W,R_1,...,R_k,P_1,$ $...,P_m\bigr)$
and a $(\Pi,\mathcal{R})$-automaton
$$A\ :=\ \bigl(Q,\mathcal{M},\pi,\delta,\mu,F,G\bigr)$$
define a
synchronized distributed computation system which executes \emph{communication
rounds} defined as follows. Each round $n\in\mathbb{N}$ defines a \emph{global configuration} $f_n:W\rightarrow Q$.
The configuration $f_0$ of the zeroth round is the function $f_0$ such that $f_0(w) = \pi(\{\ P\in\Pi\ |\ w\in P^M\ \})$
for all $w\in W$.
Recursively, assume that we have defined $f_n$,
and let $(N_1,...,N_k)$
be a tuple of message sets 
$$N_i\ =\ \bigl\{\ m\in \mathcal{M}\ |\ m = \mu( f_{n}(v), R_i),\   v\in\mathit{succ}(R_i,w)\ \bigr\}.$$
Then
$f_{n+1}(w)\ =\ \delta\bigl(\, (N_1,...,N_k),\, f_{n}(w)\, \bigr)$.
When we talk about \emph{the state of the automaton $A$ at the node $w$ in round $n$}, we mean
the state $f_n(w)$. We define that an automaton $A$ \emph{accepts} a
pointed model $(M,w)$ if there exists some $n\in\mathbb{N}$
such that $f_n(w) \in F$, and furthermore, for all $m<n$, $f_m(w) \not\in G$.
Similarly, $A$ \emph{rejects} $(M,w)$ if there exists some $n\in\mathbb{N}$
such that $f_n(w) \in G$, and for all $m<n$, $f_m(w) \not\in F$. 
Notice that $A$ may keep passing messages and changing state even after it has accepted or rejected.
Automata that stop sending messages after accepting or rejecting can be 
modelled in this framework by automata that begin sending only the message
``I have halted" once they have accepted or rejected.
(Notice that the behaviour of the distributed system does not
have to be Turing computable in any sense.)
Let $\mathcal{C}$ be the class of all pointed $(\Pi,\mathcal{R})$-models. Let $\mathcal{H}\subseteq\mathcal{C}$.
We say that $A$ accepts (rejects) $\mathcal{H}$
if the class of pointed models in $\mathcal{C}$ that $A$ accepts (rejects) is $\mathcal{H}$.
Let $\mathcal{J}\subseteq \mathcal{K}\subseteq\mathcal{C}$.
We say that $A$ accepts (rejects) $\mathcal{J}$ in $\mathcal{K}$
if the class of pointed models in $\mathcal{K}$ that $A$ accepts (rejects) is $\mathcal{J}$.
A $(\Pi,\mathcal{R})$-automaton $A$
\emph{converges} in the class $\mathcal{K}$ if for all $(M,w)\in\mathcal{K}$, the automaton $A$ either
accepts or rejects $(M,w)$.
%
%
A $(\Pi,\mathcal{R})$-automaton $A = \bigl(Q,\mathcal{M},\pi,\delta,\mu,F,G\bigr)$
\emph{halts} in $\mathcal{K}$
if $A$ converges in $\mathcal{K}$, and furthermore,
for each state $q\in F\cup G$ that is obtained by $A$
at some $(M,w)\in\mathcal{K}$, the state of $A$ at $(M,w)$
will be $q$ forever once $q$ has been obtained for the first time.
%
%
%
We say that the automaton $A$ \emph{specifies a local algorithm} in $\mathcal{K}$
if there exists some $n\in\mathbb{N}$ such that for all $(M,w)\in\mathcal{K}$,
the automaton $A$ accepts or rejects $(M,w)$ in some round $m\leq n$.
The smallest such number $n$ is called the \emph{effective running time} of $A$ in $\mathcal{K}$.
For the sake of curiosity, note that even if $A$ specifies a local algorithm, it
does not necessarily halt. However, a corresponding halting automaton of course exists.
Let $\mathcal{K}$ be a class of pointed $(\Pi,\mathcal{R})$-models.
%
%
%
%
%
%
When we say that an algorithm $A$ (or more 
rigorously, a $(\Pi,\mathcal{R})$-automaton $A$) is \emph{strongly
nonlocal} in $\mathcal{K}$, we mean that there exists no $(\Pi,\mathcal{R})$-automaton
$B$ that specifies a local algorithm in $\mathcal{K}$ and accepts
exactly the the same pointed models in $\mathcal{K}$ as $A$.
Our framework with $(\Pi,\mathcal{R})$-automata
operating on $(\Pi,\mathcal{R})$-models is rather flexible
and general. For example, each system in 
the comprehensive collection of
\emph{weak models of distributed computing}
studied in \cite{hella, hella2} can be directly
simulated in our framework by restricting attention to suitable
classes of $(\Pi,\mathcal{R})$-structures and $(\Pi,\mathcal{R})$-automata.
Let us have a closer look at this matter.
%
%
%

%
%
%
%
%
%
%
%
%
%
%
%
%
%
%

%
Let $\mathcal{R} = \{R\}$ and let $\Pi$ be any finite set.
If $M$ is $(\Pi,\mathcal{R})$-model, where $R^M$ is a symmetric and irreflexive binary relation, then
$M$ is an \emph{$\mathrm{SB}(\Pi)$-model}. The letter $\mathrm{S}$ stands
for the word \emph{set} and the letter $\mathrm{B}$ for \emph{broadcast}.
The intuition behind the framework provided by $\mathrm{SB}(\Pi)$-models is that
message passing automata see \emph{neither input port numbers nor output port numbers}.
This means that the state transition of an automaton depends only on the
current state and the \emph{set} of messages received---rather than
the multiset for example---and an
automaton must \emph{broadcast} the same message to \emph{each} of its
neighbours during a communication round.
It is not possible to send different messages to different neighbours during the
same communication round.
The framework provided by $\mathrm{SB}(\Pi)$-models is similar to
the weakest (in computational capacity) computation model $\mathrm{SB}$ studied in \cite{hella, hella2}.
In fact, the framework of $\mathrm{SB}(\Pi)$-models in the current paper is a canonical
generalization of the model $\mathrm{SB}$ in \cite{hella, hella2}.
In the article \cite{hella, hella2}, all \emph{classes} of structures
studied are always associated with a finite maximum degree bound,
and furthermore, all structures are assumed to be finite.
In the current article, such restrictions need not apply.
Also, we allow arbitrary interpretations of the unary relation symbols in $\Pi$,
while in the $\mathrm{SB}$ model of \cite{hella, hella2}, unary relation symbols
always indicate the degree of a node in a network (and nothing else).\hspace{0.4mm}\footnote{We 
do not need the to define the $\mathrm{SB}$ model
used in \cite{hella, hella2} for the purposes of the current article.
For the precise definition, see \cite{hella, hella2}. It is worth mentioning here once more, however,  
that all systems studied in \cite{hella, hella2} can be directly simulated in our framework
by simply restricting attention to suitable automata and suitable classes of pointed models.}
The reason for generalizing the definition of \cite{hella, hella2} is that in the current paper
we opt for generality as well as increased mathematical simplicity.
The philosophy in \cite{hella, hella2} is more application oriented.
%

%
%
%
%
%
%
%
%
%
%
%
%

%
Let $n\in\mathbb{N}\setminus\{0\}$ and $S = \{\, 1,...,n\, \}$.
Let $\Pi = \{\, P_0,...,P_n\, \}$ and $\mathcal{R} = \{\, R_{(i,j)}\ |\ (i,j)\in S\times S\, \}$.
A pointed $(\Pi,\mathcal{R})$-model $(M,w)$ is an \emph{$n$-port-numbering structure}, or a $\mathrm{PN}(n)$-structure,
if it satisfies the following (admittedly long and technical, and for the 
current paper rather unimportant) list of conditions.
\begin{enumerate}
\item
The union $R$ of the relations $R_{(i,j)}^M$ is a symmetric and irreflexive
relation.
%
%
\item
For any two distinct pairs $(i,j),(k,l)\in S\times S$, if $R_{(i,j)}^M(u,v)$, then $R_{(k,l)}^M(u,v)$ does not hold.
\item
For each $(i,j)\in S\times S$, if $R_{(i,j)}^M(u,v)$, then $R_{(j,i)}^M(v,u)$.
\item
For each $(i,j)\in S\times S$, the out-degree and in-degree of $R_{(i,j)}^M$ is at most one at each node.
%
%
%
%
%
\item
If $R_{(i,j)}(u,v)$ for some nodes $u$ and $v$ and some $i,j\in S$,
then, if $k<i$, there exists some $l\in S$ and some node $v'$ such that $R_{(k,l)}^M(u,v')$.
\item
Similarly, if $R_{(i,j)}^M(u,v)$ for some nodes $u$ and $v$ and some $i,j\in S$,
then, if $k<j$, there exists some $l\in S$ and some node $u'$ such that $R_{(l,k)}^M(u',v)$. 
\item
Finally, for each node $u$ and each $i\in \{0,...,n\}$, we have $u\in P_i^M$ if and only if the out-degree
(or equivalently, in-degree) of the union $R$ of all the relations in $\mathcal{R}$ is $i$ at $u$.
\end{enumerate}
It is straightforward to show that there exists a first-order $(\Pi,\mathcal{R})$-sentence
$\varphi_{\mathrm{PN}(n)}$ that defines the class $\mathcal{PN}(n)$ of all $\mathrm{PN}(n)$-structures.
This piece of information will be used in the very end of the current article when
we discuss concrete applications of Theorem \ref{maintheorem}\hspace{0.4mm}.
The class of \emph{finite}
$\mathrm{PN}(n)$-structures is \emph{exactly} the collection of communication networks of
maximum degree at most $n$ used in the framework of 
the port-numbering model $\mathrm{VV}_c$ of \cite{hella, hella2}.
The related collection of \emph{$\mathrm{VV}_c$-algorithms}
corresponds to the class of algorithms that can be specified by $\mathrm(\Pi,\mathcal{R})$-automata
that halt in all finite $\mathrm{PN}(n)$-structures.
Therefore the class $\mathcal{PN}(n)$ of 
exactly all $\mathrm{PN}(n)$-structures, together with $(\Pi,\mathcal{R})$-automata, defines a
generalization of the port-numbering model to the context with infinite structures in addition to
finite ones. Theorem \ref{maintheorem} shows that all halting algorithms for $\mathcal{PN}(n)$
are constant-time algorithms.
There are no nonlocal halting algorithms in the framework of the port-numbering model when infinite structures are
included in the picture.
%
%

%
The port numbering model $\mathrm{VV_c}$ has been studied
extensively since the 1980s. The related investigations
were originally initiated by Angluin in \cite{angluin}.
Section 3 of \cite{hella2} gives a brief and accessible
introduction to the port-numbering model and
its relation to other models of distributed computing.
%

%
%
%
%
%
%

%
%

%
%

%
\section{Halting in the Finite}\label{finite}
In this section we prove that when 
attention is restricted to finite structures, halting
and strongly nonlocal algorithms exist
even when the model of computing is defined by $\mathrm{SB}(\Pi)$-models.
While the existence of such algorithms may not be surprising,
it is by no means a trivial matter. Indeed, as we shall see in Section \ref{arbitrary}\hspace{0.4mm},
no such algorithms exist when infinite structures are included in the picture. 
Let $\Pi = \{\, P_0,P_1,Q_1,Q_2,Q_3\, \}$ and $\mathcal{R} = \{R\}$.
We will show that there exists a
strongly nonlocal algorithm that halts in the class of finite $\mathrm{SB}(\Pi)$-models.
%
%

%
%
We begin by sketching a \emph{rough} intuitive description of the algorithm.
The unary relation symbols $P_0$ and $P_1$ will be used in order to define binary words in $\{\, 0,1\}^*$
that correspond to \emph{finite walks} in $(\Pi,\mathcal{R})$-models.\hspace{0.4mm}\footnote{
A finite walk in a $(\Pi,\mathcal{R})$-model $M$ is a function from some initial segment of $\mathbb{N}$
into the domain of $M$ such that 
$\bigl(f(i),f(i+1)\bigr)\in R^{M}$ for each pair $(i,i+1)$ of indices in
the domain of $f$. A finite word $s_0...s_k = s\in\{0,1\}^*$ \emph{corresponds} to a walk $f$
iff we have $f(i) \in P_{s_i}^{M}$ for each $i \in \{0,...,k\}$.}
Each pair $(A,w)$, where $A$ is an automaton and $w$ a node,
will store a dynamically growing set of increasingly long finite binary words that correspond to 
walks that originate from $w$. The walks will be oriented by the relation symbols $Q_1$, $Q_2$ and $Q_3$
such that if a node $u$ is labelled by $Q_i$, then its successor is labelled by $Q_{p(i)}$,
where $p:\{\, 1,2,3\, \}\rightarrow\{\, 1,2,3\, \}$ 
is the cyclic permutation $1\mapsto 2\mapsto 3\mapsto 1$.
A pair $(A,w)$ will halt if it records some word $s\in\{\, 0,1\, \}^*$ that contains a
\emph{cube} as a factor, i.e.,  a word $s\ =\ tuuuv$, where 
$u$ is a \emph{nonempty} word in $\{\, 0,1\, \}^*$ and $t,v\in\{\, 0,1\, \}^*$.
Upon halting, $(A,w)$ will send an instruction to
halt to its neighbours, who will then pass the message on and also halt. Thus the halting instruction
will spread out in the connected component of $w$, causing further nodes to halt.
In addition to detecting a word with a cube factor, a globally spreading 
halting instruction can also be generated due to the detection of an undesirable
labelling pattern defined by the unary predicates in $\Pi$. For example, if a node $w$
satisfies both predicates $P_0$ and $P_1$, then the labelling pattern at $w$
is undesirable. The intuition is that then $w$ does not uniquely specify an alphabet in $\{0,1\}$,
and thereby destroys our intended labelling scheme.
Similarly, a halting instruction is generated if a violation of the cyclic permutation scheme
of the predicates $Q_1,Q_2,Q_3$ is detected.
A node accepts iff it halts in a round $n\in\mathbb{N}$ for some positive even number $n$.
Otherwise it rejects upon halting.
We shall see that the algorithm is halting and strongly nonlocal in the finite.
Strong nonlocality will follow from the existence of arbitrarily long cube-free finite words.
Indeed, there exists an infinite cube-free word, known as the
\emph{Thue-Morse sequence} (see \cite{allouche} for example).
We then define the algorithm formally.
Let us say that a node $w$
is a \emph{$Q_1$-node} if $(M,w)\models Q_1 \wedge \neg Q_2\wedge \neg Q_3$.
Similarly, $w$ is a \emph{$Q_2$}-node if
$(M,w)\models Q_2 \wedge \neg Q_1\wedge \neg Q_3$
and a \emph{$Q_3$-node} if
$(M,w)\models Q_3 \wedge \neg Q_1\wedge \neg Q_2$.
A node $w$ is \emph{properly oriented} if $w$ is a $Q_i$-node for some $i\in\{1,2,3\}$,
and furthermore, $w$ has a $Q_j$-node as a neighbour if and only if $j\in\{1,2,3\}\setminus\{i\}$.
A node $w$ is \emph{properly labelled} if it is properly oriented, and furthermore, 
either $(M,w)\models P_0\wedge\neg P_1$ or $(M,w)\models P_1\wedge\neg P_0$.
Let $\{\, 0,1\, \}^+$ denote the set $\{\, 0,1\, \}^*\setminus\{\, \lambda\, \}$, where $\lambda$ is the empty word.
Let $\mathcal{L}$ be the set of finite subsets of $\{\, 0,1\, \}^+$.
The set of states of the automaton $A$ that defines our algorithm is the set
$$\mathbb{S}\ :=\ \mathcal{L}\times\{\, 0,1\, \}
\times\{\, Q_1,Q_2,Q_3\, \}\times\{\, \mathit{run},\mathit{halt}\, \}\times\{\, 0,1\}$$
of quintuples, together with an extra finite set $H$ of \emph{auxiliary states}.
The set of messages is
$$\mathbb{M}\ :=\ \mathcal{L}\times\{\, 1,2,3\, \}\times \{\, \mathit{run},\mathit{halt}\, \}$$
of triples, together with an additional finite set $H'$ of \emph{auxiliary messages}.
We next discuss the intuition behind the definition of the states in $\mathbb{S}$.
The first set $S_1$ of a state 
$(S_1,S_2,S_3,S_4,S_5)\in\mathbb{S}$ of a node $w$ in round $n$ encodes a collection of words corresponding to
walks originating from $w$.
%
%
%
The longer the automaton computes, the longer the words in $S_1$ get.
The second and third sets $S_2$ and $S_3$ are used in order to be able to detect nodes that are not properly labelled.
The second set $S_2$ (intuitively) encodes the symbol $P\in \{P_0,P_1\}$ satisfied by the node $w$:
assuming that the labelling scheme at $w$ is 
fixed such that $(M,w)\models P_0 \wedge \neg P_1$ or
$(M,w)\models P_1 \wedge \neg P_0$, then we have $S_2 = i$ iff $w$ satisfies $P_i$.
Similarly, the third set $S_3$ intuitively encodes the symbol $Q\in \{Q_1,Q_2,Q_3\}$ such that $(M,w)\models Q$.
The fourth and fifth sets $S_4$ and $S_5$ control the halting of the node $w$.
A state $(S_1,S_2,S_3,S_4,S_5)$ is an accepting final state if $S_4 =\mathit{halt}$ and $S_5 = 0$,
and rejecting final state if $S_4 =\mathit{halt}$ and $S_5 = 1$.
The state $S_5\in \{0,1\}$ simply counts whether the current computation step is even or odd.
The set $S_1$ of a message $(S_1,S_2,S_3)$ is a
set of binary words. $S_1$ corresponds to the language recorded by the sending node.
$S_2$ encodes the label in $\{ Q_1,Q_2,Q_3 \}$ that labels the sending node.
$S_3$ is a halting instruction if $S_3 = \mathit{halt}$.
In the very beginning of the computation,
the algorithm makes use of the additional states in $H$ 
and messages in $H'$ in order to locally detect nodes that
are not properly labelled. (It is of course possible that such
nodes do not exist.)
Then, if a node $w$ is proper and $(M,w)\models P_x \wedge Q_y$, where $x\in\{0,1\}$ and $y\in\{1,2,3\}$,
the state of $A$ at $w$ in round $1$
is set to be $(\{x\},x,y,\mathit{run},1)$.
If $w$ is not proper, then the state of $A$ at $w$ in round $1$ is
set to be $(\{x'\},x',y',\mathit{halt},1)$, where $x'$ and $y'$ are fixed arbitrarily.
Let $U$ be the set of messages received by a node $w$ in some round $n+1$,
where $n\in\mathbb{N}\setminus\{0\}$.
Let $(S_1,S_2,S_3,S_4,S_5)$ be the state of $w$ in round $n$.
If $S_4 = \mathit{halt}$, then the new state is the same 
state $(S_1,S_2,S_3,S_4,S_5)$.
%
%
Otherwise the new state $(S_1',S_2',S_3',S_4',S_5')$ is defined as follows.
Let $p:\{\, 1,2,3\, \}\rightarrow\{\, 1,2,3\, \}$ be the cyclic permutation
$1\mapsto 2 \mapsto 3 \mapsto 1$.
Assume first that $U$ does not contain a tuple of the form $(X,Y,\mathit{halt})$.
Then we define 
%
%
%
%
$$S_1' = \{\ v\in\{\, 0,1\, \}^*\ |\ v = xu\text{ such that }x = S_2\text{ and }u\in T
%
%
\text{ for some }(T,p(S_3),\mathit{run})\in U\ \}.$$
%
%
%
%
We set $S_2' = S_2$ and $S_3' = S_3$. We let $S_4' = \mathit{halt}$ iff $S_1'$ contains a word with a cube as a factor.
We let $S_5' \in \{\, 0,1\, \}\setminus\{\, S_5\, \}$.
If $U$ contains a tuple of the form $(X,Y,\mathit{halt})$,
we define $(S_1',S_2',S_3',S_4',S_5') = (X',Y',Z,\mathit{halt},x)$,
where $x \in \{\, 0,1\, \}\setminus\{\, S_5\, \}$, and $X'$, $Y'$ and $Z$ are fixed arbitrarily.
Let $(S_1,S_2,S_3,S_4,S_5)$ be the state of $A$ at $w$ in round $n$, where $n\in\mathbb{N}\setminus\{0\}$.
If $S_4 = \mathit{run}$, the  message
broadcast by $A$ at $w$ in round $n+1$ is $(S_1,S_3,\mathit{run})$, and if $S_4 = \mathit{halt}$, the
message is $(X,Y,\mathit{halt})$, where $X$ and $Y$  are fixed arbitrarily.
Recall that the automaton $A$ accepts iff it halts in round $n$ for
some positive even number $n$.
The set of accepting states of the automaton $A$ is exactly the set
of states of the type $(X_1,X_2,X_3,\mathit{halt},0)$.
The set of rejecting states is the set
of states of the type $(X_1,X_2,X_3,\mathit{halt},1)$.
\begin{theorem}\label{secondmain}
Let $\Pi$ be as defined above.
There exists an $\mathrm{SB}(\Pi)$ automaton $A$ that is halting but
strongly nonlocal in the class of finite pointed $\mathrm{SB}(\Pi)$-models.
%
%
%
\end{theorem}
\begin{proof}
We shall first establish that the algorithm defined above halts in the class of finite pointed
$\mathrm{SB}(\Pi)$-models.
Assume that it does not halt in some finite model $(M,w)$.
Thus $w$ must be a proper node.
By symmetry, we may assume that $(M,w)\models Q_1\wedge \neg Q_2\wedge \neg Q_3$.
It is easy to see that for each $n\in\mathbb{Z}_+$, the node $w$ must be the first member $w_1$ of
some finite walk $(w_i)_{i\in\{1,..,n\}}$ of proper nodes that satisfy the predicates $Q_i$
in the cyclic fashion such that
$(M,w_1)\models Q_1$, $(M,w_2)\models Q_2$, $(M,w_3)\models Q_3$,
$(M,w_4)\models Q_1$,
and so on.
Therefore, since $M$ is a finite model,  the node $w$ must be the first member $w_1$ of
some infinite walk $(w_i)_{i\in\mathbb{Z}_+}$ of proper nodes that satisfy the predicates $Q_i$ in
the cyclic fashion. The infinite walk must contain a cycle.
The cycle will generate a word with a cube factor that will ultimately be detected at $w$.
Therefore the automaton at $w$ halts. This is a contradiction.
To see that the automaton is strongly nonlocal, we shall consider labelled \emph{path graphs}
that encode finite prefixes of the infinite Thue-Morse sequence.
The labelled path graphs are defined as follows.
%

%
%
%
Let $\omega$ denote the infinite Thue-Morse sequence of zeros and ones. The sequence
does not contain a cube factor. For each finite nonempty prefix $v$ of $\omega$,
let $\mathit{Path}(v)$ denote the
$(\{P_0,P_1,Q_1,Q_2,Q_3\},\{R\})$-model $M$ such that the following conditions hold.
\begin{enumerate}
\item
Note first that $\nu$ is a nonempty prefix of $\omega$, so  $\nu$ is a function 
$\nu:\{0,...,k\}\rightarrow\{0,1\}$ for some $k\in\mathbb{N}$.
The domain of the model $M$ is the set $\{0,...,k\}$.
\item
The model $M$ encodes a path graph, so for each $i,j\in\{0,...,k\}$,
we have $(i,j)\in R^M$ iff $| i- j | = 1$.
\item
$M$ encodes the finite prefix $\nu$ of the Thue-Morse sequence, so
the following conditions hold for each $i\in\{0,...,k\}$.
\begin{enumerate}
\item
We have $i\in P^{M}_0$ iff $\nu(i) = 0$.
\item
Similarly, we have $i\in P^{M}_1$ iff $\nu(i) = 1$.
\end{enumerate}
\item
Let $j\in\{1,2,3\}$.
For each $i\in\{0,...,k\}$, we have $i\in Q_j$ iff  $i = j-1$ $\mathrm{mod}$ $3$.
Thus we observe that each node of $M$, with the exception of the end nodes $0$ and $k$, is
properly labelled. (Recall the definition of proper labelling from the beginning of this section.)
\end{enumerate}
Since there exist structures $\mathit{Path}(v)$ of arbitrarily large finite sizes,
and since the Thue-Morse sequence is
cube-free, it is easy to see that our automaton $A$ is strongly nonlocal.
To see this, assume that there exists a automaton $A'$ such that
$A'$ and $A$ accept (and reject) exactly the same finite pointed $\mathrm{SB}(\Pi)$-models,
and furthermore, $A'$ specifies a local algorithm.
Let $n\in\mathbb{N}$ be the effective running time of $A'$ in the class
of finite pointed $\mathrm{SB}(\Pi)$-models. 
(We assume, w.l.o.g., that $n$ is greater than, say, $10$.)
Define the prefixes $\nu:\{0,...,5n\}\rightarrow \{0,1\}$
and $\nu':\{0,...,5n+1\}\rightarrow \{0,1\}$
of the Thue-Morse sequence.
Define the pointed models $\bigl(\mathit{Path}(\nu),\nu(3n)\bigl)$
and $\bigl(\mathit{Path}(\nu'),\nu'(3n)\bigl)$.
Consider the behaviour of our original automaton $A$ on these two
pointed models. We claim that $A$ accepts
exactly one of the two pointed models.
The halting of the pair $(A,\nu(3n))$ in the model ${Path}(\nu)$
is caused by detecting a violation of the proper labelling 
scheme  at the end point $5n$.
Similarly, the halting of $(A,\nu'(3n))$ in the model ${Path}(\nu')$
is caused by detecting a violation
at the end point $5n + 1$.
The distance from the node $\nu'(3n)$ to the node $\nu'(5n+1)$ is
exactly one step greater than
the distance from the node $\nu(3n)$ to the node $\nu(5n)$.
Thus $A$ accepts exactly one of the pointed models
$\bigl(\mathit{Path}(\nu),\nu(3n)\bigl)$
and $\bigl(\mathit{Path}(\nu'),\nu'(3n)\bigl)$.
Since the pointed models $\bigl(\mathit{Path}(\nu),\nu(3n)\bigl)$
and $\bigl(\mathit{Path}(\nu'),\nu'(3n)\bigl)$ look locally similar
to the automaton $A'$, whose effective running time is $n$,
the automaton $A'$ cannot differentiate between them.
Thus $A'$ does not halt on exactly the same finite
pointed $\mathrm{SB}(\Pi)$-models as $A$.
This is a contradiction.
\end{proof}
\section{Halting and Convergence in Arbitrary Networks}\label{arbitrary}
In this section we study a 
comprehensive collection of distributed computing models in a setting
that involves infinite networks in addition to finite ones.
We establish that \emph{every} halting distributed algorithm is in fact a local algorithm.
In fact, we show that this result relativises to any class of networks definable by a first-order theory.
The strategy of proof in this section is to first appropriately characterize acceptance and rejection
of automata in terms of definability in modal logic (see Lemma \ref{firstlemma}), and then use the compactness theorem
in order to obtain the desired end result
(see the proof of Theorem \ref{maintheorem}).
The characterizations we obtain extend the
characterizations in \cite{kuusi}.
%

%
%
%
Let $\Pi$ be a finite set of unary relation symbols, and let
$\mathcal{R}=\{\, R_1,...,R_k\, \}$ be a finite set binary relation symbols.
The set $T_0$ of $(\Pi,\mathcal{R},0)$-\emph{types} 
is defined to be the set containing a conjunction 
$$\bigwedge\limits_{P\, \in\, U} P\ \wedge\ \bigwedge\limits_{P\, \in\, \Pi\setminus U}\neg\, P$$
for each set $U\subseteq\Pi$, and no other formulae. We assume some standard bracketing and ordering of
conjuncts, so that there is exactly one conjunction for each set $U\subseteq\Pi$ in $T_0$. Note also
that $\bigwedge\emptyset\ =\ \top$. The $(\Pi,\mathcal{R},0)$-type $\tau_{(M,w),0}$
of a pointed $(\Pi,\mathcal{R})$-model $(M,w)$ is the unique formula $\varphi$ in $T_0$
such that $(M,w)\models\varphi$.
Assume then, recursively, that we have defined the set $T_{n}$ of $(\Pi,\mathcal{R},n)$-types.
Assume that $T_{n}$ is finite, and
assume also that each pointed $(\Pi,\mathcal{R})$-model $(M,w)$ satisfies \emph{exactly one} type in $T_n$.
We denote this unique type by $\tau_{(M,w),n}$. Define
\begin{multline*}
\tau_{(M,w),n+1}\ :=\ \tau_{(M,w),n} \wedge\ \bigwedge\{\ \langle R_i\rangle \tau\ |\
\tau\in T_n,\ (M,w)\models\langle R_i\rangle \tau,\ i\in\{\, 1,...,k\, \}\ \}\\
\wedge\ \bigwedge\{\ \neg\langle R_i\rangle \tau\ |\ \tau\in T_n,\ (M,w)\not\models\langle R_i\rangle\tau,
\ i\in\{\, 1,...,k\, \}\ \}.
\end{multline*}
The formula $\tau_{(M,w),n+1}$ is the $(\Pi,\mathcal{R},n+1)$-type of $(M,w)$. We assume some
standard ordering of conjuncts and bracketing, so that if two types $\tau_{(M,w),n+1}$ and $\tau_{(N,v),n+1}$
are equivalent, they are actually the same formula. We define $T_{n+1}$ to be the set 
$$\{\ \tau_{(M,w),n+1}\ |\ (M,w)\text{ is a pointed }(\Pi,\mathcal{R})\text{-model}\ \}.$$
We observe that the set $T_{n+1}$ is finite, and that for each pointed $(\Pi,\mathcal{R})$-model $(M,w)$, there exists
exactly one type $\tau\in T_{n+1}$ such that $(M,w)\models\tau$.
It is easy to show by a simple induction on modal
depth that each formula $\varphi$ of $\mathrm{ML}(\Pi,\mathcal{R})$
is equivalent to the disjunction of
exactly all $(\Pi,\mathcal{R},\mathit{md}(\varphi))$-types $\tau$ such that $\tau\models\varphi$.
Here $\tau\models\varphi$ means that for all pointed $(\Pi,\mathcal{R})$-models $(M,w)$, 
we have $(M,w)\models\tau\ \Rightarrow\ (M,w)\models\varphi$. (Note that $\bigvee\emptyset = \bot$.)
Define $\mathcal{T} := \{\, \tau\ |\ \tau\text{ is a }(\Pi,\mathcal{R},n)\text{-type for some }n\in\mathbb{N}\, \}$.
A $(\Pi,\mathcal{R})$-\emph{type automaton} $A$ 
is a $(\Pi,\mathcal{R})$-automaton whose set of states is $\mathcal{T}$.
The set of messages is also the set $\mathcal{T}$.
The initial transition function $\pi$ is defined such that the state of $A$ at $(M,w)$
in round $n=0$ is the $(\Pi,\mathcal{R},0)$-type $\tau_{(M,w),0}$.
The state transition funtion $\delta$ is defined as follows.
Let $n\in\mathbb{N}$. Let $(N_1,...,N_k)$ be a sequence of sets $N_i$ of $(\Pi,\mathcal{R},n)$-types.
Let $\tau_n$ be a $(\Pi,\mathcal{R},n)$-type.
If there exists a type 
\begin{multline*}
\tau_{n+1}\ :=\ \tau_n\ \wedge\ \bigwedge\{\ \langle R_1\rangle \tau\ |\
\tau\in N_1\ \}
\wedge\ \bigwedge\{\ \neg\langle R_1\rangle \tau\ |\ \tau\in T_n\setminus N_1\ \}\\
\vdots\\
\bigwedge\{\ \langle R_k\rangle \tau\ |\
\tau\in N_k\ \}
\wedge\ \bigwedge\{\ \neg\langle R_k\rangle \tau\ |\ \tau\in T_n\setminus N_k\ \},
\end{multline*}
we define $\delta((N_1,...,N_k),\tau_n)\, =\, \tau_{n+1}$.
Otherwise we define $\delta((N_1,...,N_k),\tau_n)$ arbitrarily.
On other kinds of input vectors, $\delta$ is also defined arbitrarily.
The message construction function $\mu$ is defined such that $\mu(\tau, R_i) = \tau$ for each $R_i$.
The sets of accepting and rejecting states can be defined differently
for different type automata.
It is easy to see that the state of any type automaton $A$ at $(M,w)$ in round $n$ is $\tau$
iff the $(\Pi,\mathcal{R},n)$-type of $(M,w)$ is $\tau$.
%

%
%
%
%
%
%
%
%
%
%
%
%

%
%
%
%
%
%
%
%
%
%
%
%
%

%
%

%
\begin{lemma}\label{firstlemma}
Let $\Pi$ and $\mathcal{R} = \{R_1,...,R_k\}$ be finite sets of unary and binary relation symbols,
respectively. Let $A$ be a $(\Pi,\mathcal{R})$-automaton.
Let $\mathcal{C}$ be the class of pointed $(\Pi,\mathcal{R})$-models.
The class $\mathcal{K}\subseteq\mathcal{C}$ of pointed models accepted by $A$ is definable
by a (possibly infinite)
disjunction $\bigvee\Phi$ of formulae of $\mathrm{ML}(\Pi,\mathcal{R})$. Also the 
class $\mathcal{J}\subseteq\mathcal{C}$ of pointed models rejected by $A$ is definable
by a (possibly infinite)
disjunction $\bigvee\Psi$ of formulae of $\mathrm{ML}(\Pi,\mathcal{R})$.
The $(\Pi,\mathcal{R},n)$-type of a pointed $(\Pi,\mathcal{R})$-model $(M,w)\in\mathcal{C}$ is in $\Phi$ iff
the automaton $A$ accepts $(M,w)$ in round $n$. Similarly, the $(\Pi,\mathcal{R},n)$-type of $(N,v)\in\mathcal{C}$
is in $\Psi$ iff the automaton $A$ rejects $(N,v)$ in round $n$.
%
%
\end{lemma}
\begin{proof}
Let $(M,w)$ be a pointed $(\Pi,\mathcal{R})$-model.
Let $B$ be a $(\Pi,\mathcal{R})$-automaton.
Let $n\in\mathbb{N}$. We let $B\bigl((M,w),n\bigr)$
denote the state of the automaton $B$ at the node $w$ in round $n$.
We shall first show that for all $n\in\mathbb{N}$ and all pointed $(\Pi,\mathcal{R})$-models $(M,w)$ and $(N,v)$,
if the models $(M,w)$ and $(N,v)$ satisfy exactly the same $(\Pi,\mathcal{R},n)$-type,
then $B\bigl((M,w),m\bigr) = B\bigl((N,v),m\bigr)$ 
for each $m\leq n$ and each $(\Pi,\mathcal{R})$-automaton $B$.
We prove the claim by induction on $n$.
For $n=0$, the claim holds
trivially by definition of the
transition function $\pi$.
Let $(M,w)$ and $(N,v)$ be pointed $(\Pi,\mathcal{R})$-models that satisfy the same
$(\Pi,\mathcal{R},n+1)$-type $\tau_{n+1}$.
Let $B$ be a $(\Pi,\mathcal{R})$-automaton and $\delta$ the transition function of $B$.
Call $q_n = B\bigl((M,w),n\bigr)$ and $q_{n+1} = B\bigl((M,w),n+1\bigr)$.
Let $N_1,...,N_k$ be sets of $(\Pi,\mathcal{R},n)$-types such that $\tau_{n+1}$ is the formula
\begin{multline*}
\tau_n\ \wedge\ \bigwedge\{\ \langle R_1\rangle \tau\ |\
\tau\in N_1\ \}
\wedge\ \bigwedge\{\ \neg\langle R_1\rangle \tau\ |\ \tau\in T_n\setminus N_1\ \}\\
\vdots\\
\bigwedge\{\ \langle R_k\rangle \tau\ |\
\tau\in N_k\ \}
\wedge\ \bigwedge\{\ \neg\langle R_k\rangle \tau\ |\ \tau\in T_n\setminus N_k\ \}.
\end{multline*}
Since the models $(M,w)$ and $(N,v)$ satisfy $\tau_{n+1}$, they must
satisfy the $(\Pi,\mathcal{R},n)$-type $\tau_n$. 
By the induction hypothesis, we therefore conclude that $B\bigl((M,w),m\bigr) = B\bigl((N,v),m\bigr)$
for each $m\leq n$. In particular, $B\bigl((N,v),n\bigr) = q_n$.
We must still show that $B\bigl((N,v),n+1\bigr) = q_{n+1}$.
Let us define that if $L$ is the set of exactly all $(\Pi,\mathcal{R},n)$-types $\tau$ such that
$(M,w)\models \langle R_i\, \rangle\tau$, then $L$ is the
\emph{set of $(\Pi,\mathcal{R},n)$-types realized by the $R_i$-successors of $w$}.
Let $i\in\{1,....,k\}$.
Since $(M,w)$ and $(N,v)$ satisfy the same $(\Pi,\mathcal{R},n+1)$-type $\tau_{n+1}$,
the set of  $(\Pi,\mathcal{R},n)$-types realized
by the $R_i$-successors of the point $w$
is the same as the set realized by the $R_i$-successors of $v$;
that set is $N_i$ in both cases. Therefore, by the induction hypothesis, the set of
states obtained by the $R_i$-successors of $w$
in round $n$ is exactly the same as the set of
states obtained by the $R_i$-successors of $v$
in round $n$. This holds for all $i\in\{1,...,k\}$.
Thus $w$ and $v$ receive exactly the same $k$-tuple of message sets in round $n+1$.
Therefore, since $B\bigl((N,v),n\bigr) = B\bigl((M,w),n\bigr) = q_n$, we conclude that
$B\bigl((N,v),n+1\bigr) = q_{n+1}$, as required.
%

%
%
%
%
%
%
%
%
%
%
%
%
%

%
We have now established that if $(M,w)$ an $(N,v)$ satisfy the same $(\Pi,\mathcal{R},n)$-type,
then any automaton $B$ produces the same state at $(M,w)$ and $(N,v)$ in all rounds $m\leq n$.
We are ready to complete the proof of the current lemma.
Let $A$ be an arbitrary $(\Pi,\mathcal{R})$-automaton.
%
%
%
Let $\mathbb{T}$ denote the set 
$$\{\ \tau\ |\ \tau\text{ is a }(\Pi,\mathcal{R},n)\text{-type for some }n\in\mathbb{N}\ \}.$$
%
%
%
%
%
%
%
Let $\Phi$ denote the set of exactly all types $\tau\in\mathbb{T}$ such
that for some $n$, the type $\tau$ is the $(\Pi,\mathcal{R},n)$-type of 
some pointed $(\Pi,\mathcal{R})$-model $(M,w)$, and furthermore, the automaton $A$
accepts $(M,w)$ in round $n$. Define the (possibly infinite) disjunction $\bigvee\Phi$.
We shall establish that for all pointed $(\Pi,\mathcal{R})$-models $(M,w)$,
we have $(M,w)\models\bigvee\Phi$ iff 
$A$ accepts $(M,w)$. 
Assume that $(M,w)\models\bigvee\Phi$. Thus
$(M,w)\models\tau_n$ for some $(\Pi,\mathcal{R},n)$-type $\tau_n$ of some pointed model $(M',w')$
accepted by $A$ in round $n$.
The models $(M,w)$ and $(M',w')$ satisfy the same $(\Pi,\mathcal{R},n)$-type $\tau_n$,
and thus $A$ produces exactly the same
state at $(M,w)$ and at $(M',w')$ in each round $l\leq n$.
Therefore $(M,w)$ must be accepted by $A$ in round $n$.
Assume that $(M,w)$ is accepted by the automaton $A$.
The pointed model $(M,w)$ is accepted in some round $n$,
and thus the $(\Pi,\mathcal{R},n)$-type of $(M,w)$ is one of the formulae in $\Phi$.
Therefore $(M,w)\models\bigvee\Phi$.
%

%
%
%
We have established that $\bigvee\Phi$ defines the class $\mathcal{K}\subseteq\mathcal{C}$.
Let $\mathcal{J}\subseteq\mathcal{C}$ be the class of pointed models
rejected by $A$.
Let $\Psi$ be the set of types $\tau\in\mathbb{T}$ such
that for some $n$, the type $\tau$ is the $(\Pi,\mathcal{R},n)$-type of 
some pointed $(\Pi,\mathcal{R})$-model $(M,w)$, and furthermore, the automaton $A$
rejects $(M,w)$ in round $n$. 
By an argument practically identical to the one above establishing that $\mathcal{K}$
is definable by $\bigvee\Phi$, one can establish that $\bigvee\Psi$
defines the class $\mathcal{J}$.
%
%
\end{proof}
%

%
%

%
\begin{theorem}[Compactness Theorem, see for example \cite{ebbinghaus}]
Assume $T$ is a set of formulae of\,  $\mathrm{FO}(\Pi,\mathcal{R})$ such that for each finite
subset $T\, '$ of\, $T$, there exists a $(\Pi,\mathcal{R})$-interpretation $(M,f)$ such that $(M,f)\models T\, '$.
Then there exists a  $(\Pi,\mathcal{R})$-interpretation $(M',f')$ such that $(M',f')\models T$.
\end{theorem}
It is a well-known immediate consequence of the compactness theorem that if $T\models\varphi$,
then there is a finite subset $T\, '$ of $T$ such that $T\, '\models\varphi$.
\begin{theorem}\label{maintheorem}
Let $\Pi$ and $\mathcal{R}$ be finite sets of unary
and binary relation symbols.
Let $\mathcal{C}$ be the class of all pointed $(\Pi,\mathcal{R})$-models. Let $\mathcal{H}\subseteq\mathcal{C}$
be a class definable by a first-order $(\Pi,\mathcal{R})$-theory.
If a $(\Pi,\mathcal{R})$-automaton converges in $\mathcal{H}$, then it
specifies a local algorithm in $\mathcal{H}$.
\end{theorem}
\begin{proof}
Assume a $(\Pi,\mathcal{R})$-automaton $A$ converges in $\mathcal{H}\not=\emptyset$.
Let $\mathcal{K}\subseteq\mathcal{H}$ be the class of pointed models accepted by $A$ in $\mathcal{H}$.
By Lemma \ref{firstlemma}\hspace{0.4mm}, there is a
disjunction  $\bigvee\Phi$ of types that defines $\mathcal{K}$ with respect to $\mathcal{H}$
and a disjunction $\bigvee\Psi$ of types that defines $\mathcal{H}\setminus\mathcal{K}$ with respect to $\mathcal{H}$.
The $(\Pi,\mathcal{R},n)$-type of a pointed $(\Pi,\mathcal{R})$-model $(M,w)\in\mathcal{H}$ is in $\Phi$ iff
the automaton $A$ accepts $(M,w)$ in round $n$. Similarly, the $(\Pi,\mathcal{R},n)$-type of $(N,v)\in\mathcal{H}$
is in $\Psi$ iff the automaton $A$ rejects $(N,v)$ in round $n$.
Let $T$  be a first-order theory that defines the class $\mathcal{H}$.
Call $X = \{\  \neg \mathit{St}_x(\varphi)\ |\ \varphi\in\Phi\ \}$ and
$Y = \{\ \neg \mathit{St}_x(\varphi)\ |\ \varphi\in\Psi\ \}$. 
Since $\bigvee\Phi$ defines $\mathcal{K}$ with
respect to $\mathcal{H}$ and $\bigvee\Psi$ defines
$\mathcal{H}\setminus\mathcal{K}$ with
respect to $\mathcal{H}$, we have $X\, \cup\, Y\, \cup\, T\models\bot$.
%
%
%
%
By the compactness theorem, there is a finite set $U\subseteq X\, \cup\, Y\cup\, T$
such that $U\models \bot$.
Let $V=U\cap X$ and $W = U\cap Y$.
Define $W^* = \{\, \varphi\in\mathrm{ML}(\Pi,\mathcal{R})\ |\ \mathit{St}_x(\varphi)\in W\, \}$, 
and define $V^*$, $X^*$ and $Y^*$ analogously.
We shall next establish that $\bigwedge W^*$ defines $\mathcal{K}$ with respect to $\mathcal{H}$.
Assume $(M,w)\in\mathcal{K}$.
Thus $(M,w)\models Y^*$,
and hence $(M,w)\models\bigwedge W^*$.
Assume then that $(N,v)\in\mathcal{H}\setminus\mathcal{K}$.
Therefore $(N,v)\models X^*$.
%
%
Since $(N,v)\in\mathcal{H}$, we have $N\models T$.
Now assume, for the sake of contradiction, that $(N,v)\models\bigwedge W^*$.
Therefore $(N,f[x\mapsto v])\models X\, \cup\, W\, \cup\, T$. 
Thus $(N,f[x\mapsto v])\models U$.
Since $U\models\bot$, we conclude that $(N,f[x\mapsto v])\models \bot$.
This is a contradiction.
We then establish that $\bigwedge V^*$
defines $\mathcal{H}\setminus\mathcal{K}$
with respect to $\mathcal{H}$.
Assume $(M,w)\in\mathcal{H}\setminus\mathcal{K}$.
Thus $(M,w)\models X^*$,
and hence $(M,w)\models\bigwedge V^*$.
Assume then that $(N,v)\in\mathcal{K}$.
Therefore $(N,v)\models Y^*$.
Since $(N,v)\in\mathcal{H}$, we have $N\models T$.
Now assume, for the sake of contradiction, that $(N,v)\models\bigwedge V^*$.
Therefore $(N,f[x\mapsto v])\models V\, \cup\, Y\, \cup\, T$. 
Thus $(N,f[x\mapsto v])\models U$.
Since $U\models\bot$, we conclude that $(N,f[x\mapsto v])\models \bot$.
This is a contradiction.
The finite sets $V^*$ and $W^*$ are negations of types. Let $\Phi'$ be the set of types whose negations are in $V^*$
and $\Psi'$ the set of types whose negations are in $W^*$.
Notice indeed that $\Phi'\subseteq\Phi$ and $\Psi'\subseteq\Psi$.
The disjunction $\bigvee\Phi'$
defines $\mathcal{K}$
with respect to $\mathcal{H}$, and the disjunction $\bigvee\Psi'$ defines $\mathcal{H}\setminus\mathcal{K}$
with respect to $\mathcal{H}$.
Let $l$ be the greatest integer $j$ such that there is a $(\Pi,\mathcal{R},j)$-type in $\Phi'\cup\Psi'$.
We claim that for each pointed model $(M,w)$ in $\mathcal{H}$, the automaton $A$ either accepts or rejects $(M,w)$
in some round $m\leq l$. To see this, let $(N,v)\in\mathcal{K}$. Thus $(N,v)\models\bigvee\Phi'$,
and hence $(M,w)\models\tau$ for some $(\Pi,\mathcal{R},i)$-type $\tau\in\Phi'$, where $i\leq l$. 
Since $\Phi'\subseteq\Phi$, we have $\tau\in\Phi$.
As we already stated in the beginning of the proof of the current theorem,
the $(\Pi,\mathcal{R},n)$-type of a pointed $(\Pi,\mathcal{R})$-model $(M,w)\in\mathcal{H}$ is in $\Phi$ iff
the automaton $A$ accepts $(M,w)$ in round $n$.
Thus the fact that $\tau\in\Phi$ implies that $(N,v)$ is
accepted in round $i$ by $A$. A similar argument applies when $(N,v)\in\mathcal{H}\setminus\mathcal{K}$.
Therefore $A$ specifies a local algorithm in $\mathcal{H}$.
\end{proof}
As we saw in Section \ref{preliminaries}\hspace{0.4mm}, each class $\mathcal{PN}(n)$ is definable by a related first-order
sentence $\varphi_{\mathrm{PN}(n)}$. Hence all halting algorithms in the port-numbering model 
are local algorithms when infinite networks are allowed.
In Section \ref{finite}\hspace{0.4mm}, we saw that finiteness gives rise to
nonlocal halting behaviour. It would be interesting to
investigate what kinds of other non-first-order properties (in addition to finiteness)
there are that lead to existence of nonlocal halting algorithms.
\section{Conclusion}
We have shown that a comprehensive variety of models of 
distributed computing cannot define universally halting nonlocal algorithms when
infinite networks are allowed.
In contrast, we have shown that in the finite, even very weak models of distributed
computing can specify universally halting nonlocal algorithms.
Our proof concerning infinite networks nicely demonstrates the potential
usefulness of modal logic in investigations
concerning distributed computing.
Our work in this article concerned \emph{anonymous networks},
i.e., networks without $\mathrm{ID}$-numbers. This choice was due
to the fact that in most natural theoretical frameworks for the
modelling of computation in infinite networks,
even the reading of local $\mathrm{ID}$s would take infinitely long, and thus synchronized
communication using $\mathrm{ID}$-numbers would be impossible.
This reasoning still leaves the possibility of investigating asynchronous
computation. A natural logical framework that can accomodate 
$\mathrm{ID}$-numbers can probably be based on 
some variant of hybrid logic (see \cite{Areces}).
Hybrid logic is an extension of modal logic with \emph{nominals};
nominals are formulae that hold in exactly one node.
 It remains open at this stage, however,
how asynchronicity should be treated.  Of course there are numerous possibilities, and
different logic-based frameworks for similar investigations exist,
but we would like to develop an approach that canonically
extends the approach introduced in \cite{hella,hella2}, developed further in \cite{kuusi},
and used in the current article.

\nocite{*}
\bibliographystyle{eptcs}
\bibliography{generic}
\end{document}